\documentclass[a4paper, 11pt]{article}

\usepackage[a4paper,margin=1.2in]{geometry}
\usepackage{amsmath,amssymb,amsthm}
\usepackage{amsfonts}

\usepackage{qtree}
\usepackage{xcolor}
\usepackage{graphicx}
\usepackage{moreverb}
\usepackage{longtable}
\usepackage[linkcolor=black,colorlinks=true,citecolor=black,filecolor=black]{hyperref}

\usepackage{algpseudocode}

\usepackage{tikz}
\usetikzlibrary {calc,positioning,shapes.misc,shapes.geometric}

\usepackage{verbatim}

\newcommand{\B}[1]{\textbf{#1}}

\begin{document}

\title{Inverse Intersections for Boolean Satisfiability Problems}
\author{Paul W. Homer}
\date{January 5th, 2025}
\maketitle

\begin{abstract}
Boolean Satisfiability (SAT) problems are expressed as mathematical formulas. This paper presents a matrix representation for these SAT problems.

It shows how to use this matrix representation to get the full set of valid satisfying variable assignments. It proves that this is the set of answers for the given problem and is exponential in size relative to the matrix.

It presents a simple algorithm that utilizes the inverse of each clause to find an intersection for the matrix. This gives a satisfying variable assignment. 

\end{abstract}

\section{Background}

Boolean Satisfiability (SAT) is a decision problem, which can be phrased as "is there at least one assignment to a set of variables that satisfies a given boolean formula?" 

SAT problems are represented as a formula containing
\emph{variables} such as $x_1, \ldots, x_n$ and the \emph{operators} $AND$ ($\land$), $OR$ ($\lor$), and $NOT$ ($\neg$). Parentheses are used for separate parts of the formula. 

For this paper we will consider Boolean formulas that are in conjunctive normal form (CNF). This will cover any k-SAT problems.

For example, a problem with 4 variables can be expressed as:

\begin{equation*}
	R = 
	(\neg x_1 \lor x_3 \lor x_4 ) 
	\land (x_2)
	\land (x_1 \lor \neg x_2 \lor \neg x_3 \lor x_4)
	\land (\neg x_2 \lor \neg x_4)
\end{equation*}

\bigskip

A \emph{literal} is either the variable $x_i$ or it's negation $\neg x_i$. A $TRUE$ value will satisfy $x_i$ while a $FALSE$ value will satisfy $\neg x_i$. 

A \emph{clause} in a k-SAT problem is defined as being a collected
set of literals combined with $OR$ operators contained within parenthesis. The clauses are  combined together by $AND$ operators. In the above example there are 4 clauses.

If there is at least one variable assignment for which at least one variable in every clause is satisfied, then there is at least one answer to the problem. If there are no such assignments for any clause, then the problem is \emph{unsatisfiable}. 

SAT problems are the first known examples of NP-complete problems, and
have been proven to be polynomial reducible \cite{Cook} to 3-SAT problems. 

\newpage

\section{Problem Representations}

For the k-SAT problem:

\begin{equation*}
	R = 
	(\neg x_1 \lor x_3 \lor x_4 ) 
	\land (x_2)
	\land (x_1 \lor \neg x_2 \lor \neg x_3 \lor x_4)
	\land (\neg x_2 \lor \neg x_4)
\end{equation*}

There are $V$ variables in a k-SAT problem, for the above example $V = 4$: $x_1$, $x_2$, $x_3$, and $x_4$. They are joined by $OR$ ($\lor$).

There are $C$ clauses, in the above example there are $C = 4$ clauses joined by $AND$ ($\land$).

\subsection{Matrix Representation}

We can represent any such problem $R$ as a matrix $R_M$ of size $V*C$. This works for general problems, but also for any k-SAT problems, such as 3-SAT. 

Each variable in the problem $R$ is a row in the matrix $R_M$.  Each clause in the problem $R$ is a column in the matrix $R_M$.

For each cell in the matrix, we will use $T$ for the literal if it is $x_i$ in the formula. We will use $F$ for it's complement $\neg x_i$. 

We will use $U$ for any variable in $R_M$ that is unassigned in the current clause.

This gives us a matrix representation of $R$ as:

\begin{equation}
	R_M = 
	\begin{bmatrix}
		\B{F}  & U       & \B{T}   & U \\
		U        & \B{T} & \B{F}   & \B{F} \\
		\B{T}   & U       & \B{F}   & U \\
		\B{T}  & U       & \B{T}   & \B{F} \\
	\end{bmatrix}
\end{equation}

\bigskip

The variables can be in any order and still represent the same problem. The clauses can also be in any order. 

For any algorithm that takes a k-SAT problem as input in this format the input size $N = V*C$.

We can address individual clauses in the matrix, such as the first one in $R_M$ above:

\begin{equation*}
	C_1 = 
	\begin{bmatrix}
		\B{F}  \\
	    U  \\
		\B{T}  \\
		\B{T}  \\
	\end{bmatrix}
\end{equation*}

\subsection{Binary Tree}

We can consider a binary tree created from the variables in their order in the matrix, where each node is labelled $T$ or $F$. 

\bigskip

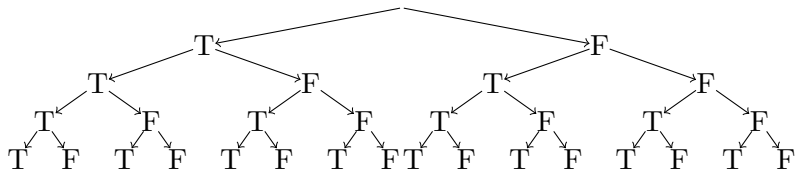
\tikzset {
	treenode/.style={align=center,inner sep=0pt},
    simple/.style = {treenode,rectangle,text width=0.5cm,text height=0.5cm},
	stock/.style = {treenode,circle,draw=black,text width=0.8cm},
    strong/.style = {treenode,circle,draw=black,text width=0.8cm}
}

\begin{center}
\begin{tikzpicture}[->,
	level 1/.style={sibling distance=52mm, level distance=5mm},
	level 2/.style={sibling distance=28mm, level distance=5mm},
	level 3/.style={sibling distance=14mm, level distance=5mm},
	level 4/.style={sibling distance=7mm, level distance=5mm}
]
\node[treenode] {}
	child { node[treenode] {T}
		child { node[treenode] {T}
			child { node[treenode] {T}
				child { node[treenode] {T} }
				child { node[treenode] {F} }
			}
			child { node[treenode] {F}
				child { node[treenode] {T} }
				child { node[treenode] {F} }
			}
		}
		child { node[treenode] {F}
			child { node[treenode] {T}
				child { node[treenode] {T} }
				child { node[treenode] {F} }
			}
			child { node[treenode] {F}
				child { node[treenode] {T} }
				child { node[treenode] {F} }
			}
		}
	}
	child { node[treenode] {F}
		child { node[treenode] {T}
			child { node[treenode] {T}
				child { node[treenode] {T} }
				child { node[treenode] {F} }
			}
			child { node[treenode] {F}
				child { node[treenode] {T} }
				child { node[treenode] {F} }
			}
		}
		child { node[treenode] {F}
			child { node[treenode] {T}
				child { node[treenode] {T} }
				child { node[treenode] {F} }
			}
			child { node[treenode] {F}
				child { node[treenode] {T} }
				child { node[treenode] {F} }
			}
		}
	}
;
\end{tikzpicture}
\end{center}

\bigskip

The depth of this tree is $V+1$. For the above example, there are 4 variables. 

\subsection{Clause Paths}

We can express any clause in $R_M$ as a set of paths through the above binary tree. 

For any given clause $C_i$, a variable assigned to $T$ goes through the right side of the node. $F$ goes through the left side, while $U$ goes through both sides. Although we use $U$ for unassigned in any give clause, we treat it as \emph{union} for paths.

In this way a clause represents a set of paths through this tree. 

For the clause $C_1 = [ F, U, T, T ]$ above we have 2 paths in the tree:

\bigskip

\begin{center}
\begin{tikzpicture}[-,
	level 1/.style={sibling distance=56mm, level distance=13mm},
	level 2/.style={sibling distance=28mm, level distance=12mm},
	level 3/.style={sibling distance=14mm, level distance=11mm},
	level 4/.style={sibling distance=7mm, level distance=10mm}
]
\node[simple] {$C_1$}
	child { node[simple] {T}
		child { node[simple] {T}
			child { node[simple] {T}
				child { node[simple] {T} }
				child { node[simple] {F} }
			}
			child { node[simple] {F}
				child { node[simple] {T} }
				child { node[simple] {F} }
			}
		}
		child { node[simple] {F}
			child { node[simple] {T}
				child { node[simple] {T} }
				child { node[simple] {F} }
			}
			child { node[simple] {F}
				child { node[simple] {T} }
				child { node[simple] {F} }
			}
		}
	}
	child[->] { node[stock] {F}
		child { node[stock] {T}
			child { node[stock] {T}
				child { node[stock] {T} }
				child[-] { node[simple] {F} }
			}
			child[-] { node[simple] {F}
				child { node[simple] {T} }
				child { node[simple] {F} }
			}
		}
		child { node[stock] {F}
			child { node[stock] {T}
				child { node[stock] {T} }
				child[-] { node[simple] {F} }
			}
			child[-] { node[simple] {F}
				child { node[simple] {T} }
				child { node[simple] {F} }
			}
		}
	}
;
\end{tikzpicture}
\end{center}

\bigskip

So, the 9th and 13th leaf nodes are part of the paths included by the clause $C_1$.

A clause with $q$ cells set to $U$ includes $2^q$ paths. So, every clause $C_i$ has $1, 2, 4, ..., 2^{V-1}$ paths depending on the number of variables set to $U$. 

A clause with all $V$ variables set to $U$ is the empty clause which is \emph{unsatisfiable}, there are no arrangement of variables that could ever satisfy it. Any k-SAT problem containing such a clause is unsatisfiable.

\subsection{Membership Sets}

For each of the $2^V$ leaf nodes at the bottom of the tree for $R_M$,  we can use a $0$ or $1$ to indicate whether the path through the tree to that leaf node is from a clause $C_i$ that is included the matrix $R_M$. 

We can combine these booleans as an ordered list that matches the arrangement of the tree's leaf nodes. This can be described as a \emph{bit string}.

If we look at all leaf nodes,  in order,  for a given column, they form a bit string which represent the set of all paths in the clause

For the above tree for the clause $C_1 = [ F, U, T, T ]$ we get a string for the membership set of:

\begin{center}
$C_1$ = 0000 0000 1000 1000 
\end{center}

This shows the two paths through the tree that are created by $C_1$. The 9th and 13th bits are set. All others are 0.

We can do that for all clauses in the k-SAT problem $R_M$:

\begin{center}
$C_1$ = 0000 0000 1000 1000 \\
$C_2$ = 1111 0000 1111 0000 \\
$C_3$ = 0000 0010 0000 0000 \\
$C_3$ = 0000 0101 0000 0101 
\end{center}

Because all of these clauses are overlaid on the same binary tree for any given problem $R_M$, we can combine them together with $OR$ to get the full set of all included paths in $R_M$.

\begin{equation*}
\begin{split}
R_S &= C_1 \lor C_2 \lor C_3 \lor C_4 \\
	&= \text{ 0000 0000 1000 1000 } \\
    	&\lor \text{ 1111 0000 1111 0000 } \\ 
		&\lor \text{ 0000 0010 0000 0000 } \\
		&\lor \text{ 0000 0101 0000 0101 } \\ 
	&= \textbf{1111 0111 1111 1101 }
\end{split}
\end{equation*}

This shows us that all paths through the tree are included in $R_S$ except for the 5th and 15th one. These two paths are not covered by any of the clauses.

\subsection{Answer Sets}

If we reverse the order of the membership set $R_S$ and then flip the bits, this is the set of all possible answers $R_A$ to the problem $R$.

\begin{center}
$R_A$ = 0100 0000 0001 0000
\end{center}

There are two possible answers to $R$ located in the 2nd and 12th position. From the positions we can find the paths they represent, which gives us their variable assignments:

\begin{equation*}
R_A =  
\left\{ 
\begin{bmatrix}
T  \\
T  \\
T  \\
F \\
\end{bmatrix},
\begin{bmatrix}
F  \\
T  \\
F  \\
F \\
\end{bmatrix}
\right\}
\end{equation*}

We can validate that both of these answers are correct by seeing that for each answer there is at least one variable assignment that will satisfy each of the clauses in $R_M$, and thus $R$. There are no other valid answers to $R$.

\subsection{Validity}

To prove this relationship between the membership set and answer sets is correct for any matrix we will start by addressing an ambiguity in the matrix representation:

\newtheorem{initial}{Lemma}[section]

\begin{initial}
For nearly identical clauses $C_i$ and $C_j$, iff all other variables are identical, contradictory variables ($x_i$, $\neg x_i$) can cancel each other out and be set to unassigned $U$, without changing the answer set of the problem $R$.
\end{initial}

If we have two clauses in a problem $R$ that are almost identical to one another but opposite in only one cell, they can be combined into a $U$. 

\begin{equation*}
	R = 
	\begin{bmatrix}
		\B{T}  & \B{T}\\
	    \B{T}  & \B{T}\\
		\B{T}  & \B{T}\\
		\B{T}  & \B{F}\\
	\end{bmatrix}
	= 
	\begin{bmatrix}
		\B{T}  \\
	    \B{T}  \\
		\B{T}  \\
		U  \\
	\end{bmatrix}
\end{equation*}

The two membership sets are the same:

\begin{equation*}
\begin{split}
C_1 &= 1000 0000 \\
C_2 &= 0100 0000 \\
& \\
C_3 &= 1100 0000
\end{split} 
\end{equation*}

even though the formula representations for the problem are different:

\begin{equation*}
\begin{split}
	R &= (x_1 \lor x_2 \lor x_3 \lor x_4 ) \land (x_1 \lor x_2 \lor x_3 \lor \neg x_4) \\
	&\\
    &= (x_1 \lor x_2 \lor x_3)
\end{split}
\end{equation*}

\begin{proof}
As all other variables of the two clauses are the same, any assignment to $x_4$ above is unnecessary. An assignment to $x_1$, $x_2$, or $x_3$ will satisfy both clauses, we don't need to assign $x_4$ for just one, so it is the same as being unassigned. There is at least one common literal to satisfy both clauses.

This is true even if all other variables \textit{except one} are $U$. If all other variables are $U$ then the problem is unsatisfiable.

So the two different matrix representations of $R$ are equivalent.
\end{proof}

With that we can show that the reversed inverted membership set is the answer set.

\newtheorem{string_inverse}{Theorem}[initial] \label{theorem1}

\begin{string_inverse}
For a problem $R$, we get a membership string $R_S$ from the paths created by each clauses in a matrix $R_M$. 

The reverse inverted string $R_A$ is the set of all valid solutions to $R$.
\end{string_inverse}

\begin{proof}
By Induction

We can define an $f(s)$ that will take a bit string $s$, reverse it's order and flip all of the bits from $1 \rightarrow 0$ and $0 \rightarrow 1$. This will be distributive with respect to the union operator.

\bigskip
\emph{Base Case}: 
Show that for 1 clause then $f(R_{S_1}) = R_{A_1}$.  

Reversing the string is equivalent to flipping any literals in the clause from $T \rightarrow F$ or $F \rightarrow T$. Flipping the bits in the reversed string is equivalent to taking the complement of the set which gives the $2^V$ possible answers.

For a single path it is obvious. For example with $[ T, T, ..., T]$, then the flipped clause $[F, F, ..., F]$ is the only assignment that isn't valid. For $[ F, T, F, ..., T, F ]$ then flipped clause $[ T, F, T, ..., F, T ]$ is the only invalid assignment. 

For $2^q$ paths in one clause with $q$ assignments to $U$, we see that each path has 1 invalid permutation from above. As they are discrete, then there are $2^q$ invalid permutations, as expected, so are $2^V - 2^q$ answers. 

So, $f(R_{S_1}) = R_{A_1}$ as expected.

\bigskip
\emph{Induction hypothesis}: 

Assume that $f(R_{S_k}) = R_{A_k}$ and that it has $2^p$ paths. That $f(s)$ is distributive for the operator $\cup$.

\bigskip
\emph{Induction step}: 

Show that $f(R_{S_{k+1}}) = R_{A_{k+1}}$ iff $f(R_{S_k}) = R_{A_k}$  and $f(R_{S_1}) = R_{A_1}$. 

With the definitions:
\begin{equation*}
\begin{split}
R_{S_{k+1}}&=R_{S_k} \cup R_{S_1} \\
R_{A_{k+1}} &= R_{A_k} \cup R_{A_1}
\end{split}
\end{equation*}

\bigskip
We see that:

\begin{equation*}
\begin{split}
f(R_{S_{k+1}})&=f(R_{S_k} \cup R_{S_1}) \\
&=f(R_{S_k}) \cup f(R_{S_1}) \\
&=R_{A_k} \cup R_{A_1} \\
&=R_{A_{k+1}} 
\end{split}
\end{equation*}

Which is the result we wanted. The paths for $\lvert R_{A_{k+1}} \rvert \leq 2^V - (2^q+2^p)$. So the answers are preserved or reduced as we add more clauses. 
\end{proof}

\section{Solving k-SAT Problems}

The size of the membership set $R_S$ and the size answer set $R_A$ are both $2^V$ which is clearly exponential relative to the input $R_M$. 

If we used the clauses in $R$ to create the full sets for the membership and answers that would be at least $O(2^N)$ to create or manipulate them. 

As any given clause $C_i$ can have bits anywhere in the membership string $R_S$, we have to consider all clauses in $R_M$ to arrive at any valid answer. Any sort of brute force search would require $O(2^N)$.

If we try to work through the sets $R_S$ or $R_A$ they will fragment into an exponential number of pieces because they alternate and overlap. The worse case is $O(2^N)$.

We can \emph{normalize} the clauses so that they are distinct and do not overlap with each other. In general, a normal form of $R_M$ with $q$ answers would have between $V-1$ and $V^2/2$ clauses, depending on $q$.

But as we try to normalize the clauses $C_i$ we run into exponential number of fragments. Depending on the positions of the $U$ assignments, we can't easily resolve overlaps between clauses without first fragmenting them.

As the clauses $C_i$ are intertwined, any sort of algorithm to use them directly to find missing entries in the membership is intertwined and causes exponential work. We can't decompose this into independent sub-problems.

So we can find the solution, but doing so directly it is an exponential effort.  

However, the clauses $C_i$ of $R_M$ are essentially polynomial representations of the underlying exponential membership $R_S$ and answer sets $R_A$. We can leverage this property for faster searching for valid answers.

We just need to do it indirectly. 

\subsection{Clause Overlaps}

The intersection between any two clauses $C_i$ and $C_j$ is a path is either empty $\varnothing$ or it is a single clause $c_k$. 

Each clause represents a bunch of paths through the tree. The paths can only be of sizes that are powers of 2. 

The paths are all connected, it is always one valid subtree of the problem tree. Any such subtree can be represented as a clause $c_k$.

\subsection{Clause Inverse}

For any clause $C_i$ we can take it's inverse $I_i$. The inverse is the set of all paths that are $not$ covered by this clause. We can flip each literal in every clause in $I_i$ to get the full set of answers for $C_i$ as if it were a stand-alone problem.

The inverse $I_i$ of any clause is up to $V$ disjoint clauses. There is one inverse clause for each literal in $C_i$. There are no inverse clauses for variables set to $U$. 

We can easily calculate the inverse $I_i$ by going through each variable in $C_i$, if there is a literal, the inverse is the other non-included subtree. These subtrees are added to the inverse set of clauses.

So if we have a path for $C_i$ with two literals, then the inverse is ($I_{i_1}$, $I_{i_2}$) are the subtrees off of this:

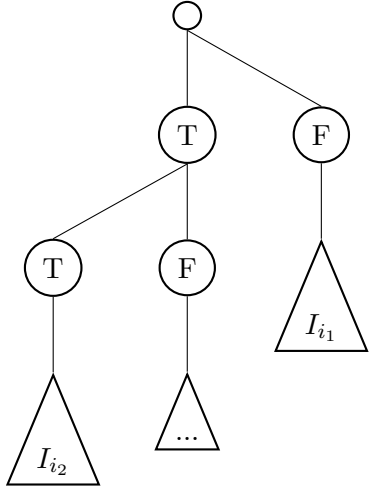
\begin{center}
\begin{tikzpicture}
[
	simple/.style={circle,draw=black,thick},
	subtree/.style={isosceles triangle,draw=black,thick,shape border rotate=90}
]
	
\node[simple] (root) { } ;
\node[simple] (child1) [below=of root] {T} ;
\node[simple] (alt1) [right=of child1] {F} ;
\node[simple] (child2) [below=of child1] {F} ;
\node[simple] (alt2) [left=of child2] {T} ;

\node[subtree] (sub1) [below=of alt1.south] {$I_{i_1}$} ;
\node[subtree] (sub2) [below=of alt2.south] {$I_{i_2}$} ;
\node[subtree] (other) [below=of child2] {...} ;

\draw[-] (root.south) -- (child1.north) ; 
\draw[-] (root.south) -- (alt1.north) ;
\draw[-] (child1.south) -- (child2.north) ;
\draw[-] (child2.south) -- (other.north) ;
\draw[-] (alt1.south) -- (sub1.north) ;
\draw[-] (alt2.south) -- (sub2.north) ;
\draw[-] (child1.south) -- (alt2.north) ;
\end{tikzpicture}
\end{center}

The inverses are disjoint, so normalized by construction. 

\subsection{Inverse Intersections}

For a group of inverses in $R_M$, as they are effectively combined with $AND$, as per Theorem \ref{theorem1}, all we need to do is find the intersection between all of these different possible answer sets.

To be valid, the same answer (variable assignment, path) must appear in all of the inverses $I_i$.

We can optimize finding these intersections because the inverse clauses are disjoint and each clause itself is a polynomial set of possible answers. 

The intersection between any two clauses $C_i$ and $C_j$ is another clause $d_k$, but the intersection between a clause $C_i$ and an inverse $I_i$ is a set of clauses. 

As an inverse $I_i$ is a set of disjoint clauses, comparing that to a clause may produce an overlap for any of the inverse clauses $I_{i_c}$. If there are $k$ clauses in $I_i$ then there can be between $0$ and $k$ intersections.

These $k$ intersection clauses are disjoint because all clauses of $I_i$ are disjoint.

So, for a single clauses there are $k$ intersections, and for all clauses for $I_i$ and $I_j$ there could be $k*l$ intersections. A direct comparison of all intersections with each other would produce a large number of intersections, which would have to be directly compared, etc. so it is exponentially fragmenting if we are not careful.

\subsection{Considerations}

We want to quickly find the intersection between any large set of inverses. 

We can accomplish this by taking all of the clauses for all inverses $I_i$, and treating them as a \emph{candidate} $c_k$ of potential answers to $R_M$. 

We compare each candidate $c_k$ against all of the other inverses $I_j$ to find any overlaps.

When we find an intersection between the candidate and any other clause, we can treat that overlap, which is smaller, as a new candidate and continue from where we are. 

If there are two or more overlaps between the clauses for a single inverse $I_i$, we have to continue the extra clause testing independently, although we can start where the split occurred. 

Since all clauses in an inverse are disjoint, the candidates will be disjoint. We do not want to lose any potentially correct answers as they may be the only answers to the problem.

\newtheorem{lost}{Lemma}[initial]

\begin{lost}
For $R_M$, the intersection of all of its inverses $I_i$ is $R_A$.

\begin{center}
$R_A = I_1 \cap I_2 \cap ... \cap I_C$
\end{center}

Then no valid possible answers are lost from the set of answers $R_A$ by taking the intersection between all of the inverses $I_i$.
\end{lost}

\begin{proof}
Relative to any one clause $c_k$ in $I_i$ its answer set is its full set of answers. But they are not necessarily an answer to $R_M$ unless they is also included in all other inverses $I_j$.

As the answer sets are disjoint, we can compare any one clause in an inverse $I_i$ against all other clauses in another inverse $I_j$. There may be $0$, $1$, or $p$ intersections. 

\bigskip 

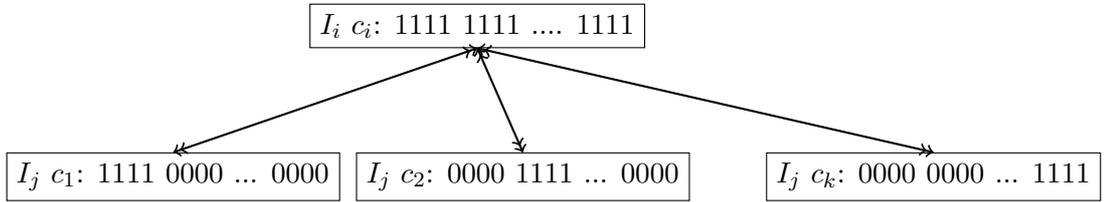
\begin{tikzpicture}
\path (4,2) node(a) [rectangle,draw] {$I_i$ $c_i$: 1111 1111 .... 1111} 
(0,0) node(b) [rectangle,draw] {$I_j$ $c_1$: 1111 0000 ... 0000} 
(4.6,0) node(c) [rectangle,draw] {$I_j$ $c_2$: 0000 1111 ... 0000} 
(10,0) node(d) [rectangle,draw] {$I_j$ $c_k$: 0000 0000 ... 1111} ;
\draw[thick,<<->>] (a.south) -- (b.north) ;
\draw[thick,<<->>] (a.south) -- (c.north) ;
\draw[thick,<<->>] (a.south) -- (d.north) ;
\end{tikzpicture}

\bigskip 

The intersection between a candidate $c_k$ and an inverse $I_j$ is a set of clauses, each of which is a smaller candidate. 

For any bits that were dropped from $c_k$, there is at least one inverse $I_j$ where they are not a valid answer.

For the bits that were dropped on the inverse side $I_j$, eventually they will be candidates and compared back to all other inverses. They will be tested, we will not lose them as possible answers.
\end{proof}

With this, we can find the intersection of these inverses $I_i$ with minimum effort.

We can start by taking all of the clauses $c_k$ in all of the inverses $I_i$ as possible candidates. As they are disjoint with each other in their inverse $I_i$, we can treat them separately as just a set of possible candidates. 

We calculate the intersection of any two clauses. It is empty or a single clause.

We calculate the intersection of a clause $c_k$ and any inverse $I_j$ by going through each of its clauses. The clause $c_k$ may intersect with zero, one or more of the inverse clauses, so we get a set of intersections that are candidates.

We can take any clause in the set of intersections and use it as a potential candidate for 1 or more answers.

As we compare $c_k$ to all other inverses, the size of the answer set may be reduced. It is always reduced exponentially. It will decrease by $1/2$, $1/4$, $1/8$, ... $1/2^q$ of the bits. 

If an inverse clause of any size survives comparison to all other inverses $I_i$, then it is a set of answers to the problem $R$. If there are any remaining $U$ settings, we can pick any literal for each one and then return a valid answer.

Since, an inverse may contain up to $V$ clauses, if there is an intersection and the intersection is smaller than the original clause, we have to consider each V overlap.

Thus if we test $c_1$ against all of the inverses $I_j$, it will get reduced. If it disappears entirely, it is not a valid candidate. If any of the possible answers makes it to the end, then they are valid.

While testing $c_1$, we will also generate fragments, $d_1$, $d_2$, etc. We can ignore these if they are identical to $c_1$. If not, they are exponentially smaller. 

So, although we are splitting off into new candidates with each test, either the candidate remains the same, or it is guaranteed to be exponentially smaller.

\subsection{Inverse Intersection Algorithm}

We can create a simple algorithm $A$ that will accomplish the goals above. To make it easier to see the computational complexity we will \emph{queue} the candidates first, then test them and add new candidates to the queue.

First we calculate all of the inverses and queue each of their clauses.

Second, we take each candidate in place it in a queue, and compare it all of the inverses. 

If we find a single overlap between the candidate and one of the inverse clauses, we use that reduced clause as the new candidate. 

If we find multiple overlaps between different clauses in the inverse, we use the first one as above, and then queue the second one as a new candidate.

\bigskip

\begin{algorithmic}[1]
\Procedure{Intersection}{$r$} \Comment{find intersection for inverses}
	\ForAll{clauses in problem} \Comment{$O(C)$}
		\State $inverse = Inverse(clause)$
		\ForAll{clauses in inverse} \Comment{$O(V)$}
			\State Add to Queue
		\EndFor
	\EndFor
	\While{Queue not empty} \Comment{O(C)*O(V)+O(overlaps)}
    	\State Get next $candidate$
		\For{each inverse} \Comment{O(C)}
			\For{each clause in inverse} \Comment{O(V)}
				\State $overlap = Overlap(candidate, clause)$
				\If {overlap is nil}
					\State $break;break;$ \Comment{Terminate both for loops}
				\Else
					\If {first}
						\State Replace the $candidate$ with $overlap$
					\Else
						\State Add $overlap$ to Queue
					\EndIf
				\EndIf
			\EndFor
			\If {no overlaps in inverse found}
				\State \textbf{return} nil
			\EndIf				
		\EndFor
		\If {candidate not nil}
			\State \textbf{return} $candidate$ \Comment{Paths intersect with all inverses}			
		\EndIf
	\EndWhile
	\State \textbf{return} nil \Comment{No answer, nil for unsat}
\EndProcedure
\end{algorithmic}

\bigskip

We need to find the inverse of any clause. It will be a set of  between $1$ and $V-1$ clauses.

We go through any clause, variable by variable. If we find a literal, then 
we create a new clause with the same literal settings above and all $U$ values below. We flip the current index.

So it is the flipped literal, and a subtree of all $U$ assignments.

\bigskip

\begin{algorithmic}[1]
\Function{Inverse}{$clause$}\Comment{inverse a clause}
	\State set inverse to all U
	\ForAll{(index,variable) in clause}\Comment(O(V))
		\If{variable is T}\Comment{Handle T or F, but ignore U}
			\State inverse[index] = F
			\State Add inverse to Results
			\State inverse[index] = variable \Comment{Reset it for the next possible inverse}
		\ElsIf {variable is F}
			\State inverse[index] = T
			\State Add inverse to Results
			\State inverse[index] = variable \Comment{Reset it for the next possible inverse}
		\EndIf
	\EndFor
	\State return Results \Comment{There is always at least one inverse clause}
\EndFunction
\end{algorithmic}

\bigskip

We need to find the overlap between any two clauses.

We go through both of the clauses at the same time. If the variables are the same, they are included in the new clause. If they contradict each other T != F, then there is no overlap, we can return right away.

If one side is a literal and the other is unassigned, then we pick the literal.

\bigskip

\begin{algorithmic}[1]
\Function{Overlap}{$c1$,$c2$}\Comment{find overlap between two clauses}
	\ForAll{index,variables in clause c1}\Comment{O(C)}
		\If{c1[index] == c2[index]} \Comment{identical T,F, or U, copy it over}
			\State $r[index] = c1[index]$
		\ElsIf{c1 == U and c2 != U} \Comment{right is literal, use that}
			\State $r[index] = c2[index]$	
		\ElsIf{c1 != U and c2 == U} \Comment{left is literal, use that}
			\State $r[index] = c1[index]$	
		\ElsIf{c1 != c2 $and$ both are literals} \Comment{T != F, no overlap}
			\State return nil
		\Else
			\State $r[index] = c1[index]$	\Comment{Copy the left side}
		\EndIf
	\EndFor
	\State return r
\EndFunction
\end{algorithmic}

\bigskip

As well, we need to parse input problems and convert them into matrix format. 
Then we can apply Inverse and see if there is an answer is produced or not. 

\subsection{Computational Complexity}
From the way it is constructed, the complexity of the algorithm $A$ is entirely dependent on the growth of the candidate queue. 

\newtheorem{growth}{Theorem}[initial]

\begin{growth}
The queue and testing for algorithm $A$ grows no faster than $O(K^N)$.

Although it is potentially adding new candidates for each test against each inverse, the size candidates are shrinking exponentially which is faster than the queue is growing.
\end{growth}

\begin{proof}

An inverse can have up to $V$ clauses. A candidate may intersection with all of them. But they are all disjoint, so the intersections themselves will not overlap. 

If a clause has $q$ intersections, then each intersection can must be at least $1/2^q$ in size.

But each intersection can only be split $V-q$ more times at most.

So if a candidate fragments and the reductions will get smaller exponentially.

We can show that the candidates get smaller much faster than the splits grow.

To do this we can look at the number of intersections that may occur for each inverse. 

\bigskip

\emph{Intersection Cases}:

If all candidates have exactly 0 intersections, then we test them all at once, and the queue does not grow. Testing is $V$ so we have $O(N*V) \sim O(N^2)$

If all candidates have exactly 1 intersection, we can substitute that for the candidate itself, so it is identical to the 0 intersection case. 

If all candidates have exactly the maximum V intersections, that is each one has a unique intersection with every clause in the inverse, then the queue will grow by $N*V$. But each intersection will have been split V times, which means that it's size is one. So, we'll get the initial $N*V$ tests, plus a second generation of $(N*V)*V$ tests, which again is $~O(N^2)$

We know in between these edge cases, that the intersections can be exponential making the queue grow like a tree. 

If all candidates have 3 intersections for example, we can substitute 1 away, but 2 more end up in the queue. Which forms a binary tree of $3^{g(x)}$. However, we also know that the intersections are disjoint, so the worst case for overlaps are sizes $(1/2, 1/4, 1/8)$.  If we have a maximum depth of $V$ for the tree, then splitting by $2$, for example, means we go down each time at $i/2$ where $i = V..1$, which means the number of steps until size 1 is $1/3 * V$. Since we are dividing by $4$ for this case, it is $4/3*V$. Then the queue size is $N + N *3^{l4/3*V}$. This is $\sim O(3^N)$ for this case.

We can see that this applies to $q$ intersections as well. We get a tree with $q$ branches. We get a divisor of $1/2^q$ so a depth of $1/q*V$. So we land on $q^1/q*V$.

For any $q$, the overlap sizes may not be evenly split into $q$ pieces. But since they are disjoint, if one piece is larger than $1/2^q$ there must be a corresponding piece that is smaller. The smaller pieces will cancel out the larger ones.
	
We know that the number of overlaps will not be exactly $0, 1, 3, q,$ or $V$ but will vary between these at each level. 

We can see that for a strict k-SAT problem, q = k. If the problem has clauses with varying numbers of literals up to $V$, then it is likely that the worse case is q as the average of all of these.
\end{proof}

The algorithm can solve k-SAT problems, including 3-SAT. It does so in a worst case of $O(k^N)$. The space size of the queue is bounded by $O(k^N)$. Although the algorithm is exponential, the underlying objects it is working with are polynomial and are shrinking faster than the queue is growing. The exponential worst case growth is then cancelled out by the logarithmic reductions.



\end{document}